\documentclass[conference]{IEEEtran}

\usepackage[utf8]{inputenc}
\usepackage[pdftex]{graphicx}
\usepackage{subcaption}
\usepackage[draft]{todonotes}
\usepackage{booktabs}
\usepackage{array}
\usepackage{amsmath}
\usepackage{bm}
\usepackage{amsfonts,amssymb}
\usepackage{cleveref}
\usepackage{cite}
\usepackage{xspace}
\usepackage{wrapfig}
\usepackage{subcaption}
\usepackage{soul}
\usepackage{chngcntr}
\usepackage{url}
\usepackage{amsthm}

\usepackage[font=scriptsize,labelfont=bf]{caption} 

\newcommand\sysname{Proof of Prestige\xspace}

\newcommand\pow{PoW\xspace}
\newcommand\pos{PoS\xspace}
\newcommand\pop{PoP\xspace}

\newcommand\blockchain{blockchain\xspace}

\newcommand{\eg}{\textit{e.g.,}\@\xspace}
\newcommand{\ie}{\textit{i.e.,}\@\xspace}

\def\first{({\it i})\xspace}
\def\second{({\it ii})\xspace}

\newcolumntype{L}{l<{\hspace{1cm}}}
\newcolumntype{C}{c<{\hspace{.3cm}}}

\newtheorem{theorem}{Theorem}

\renewcommand{\baselinestretch}{0.94}

\begin{document}


\title{Proof-of-Prestige: A Useful Work Reward System for Unverifiable Tasks}

\author{Michał Król, Alberto Sonnino, Mustafa Al-Bassam, Argyrios Tasiopoulos, Ioannis Psaras\\University College London, United Kingdom}

\maketitle

\IEEEpubidadjcol

\begin{abstract}

As cryptographic tokens and altcoins are increasingly being built to serve as utility tokens, the notion of useful work consensus protocols, as opposed to number-crunching PoW consensus, is becoming ever more important. In such contexts, users get rewards from the network after they have carried out some specific task useful for the network. While in some cases the proof of some utility or service can be proved, the majority of tasks are impossible to verify.
In order to deal with such cases, we design “Proof-of-Prestige” (PoP)---a reward system that can run on top of  Proof-of-Stake blockchains. PoP introduces “prestige” which is a volatile resource and, in contrast to coins, regenerates over time. Prestige can be gained by performing useful work, spent when benefiting from services and directly translates to users minting power.
PoP is resistant against Sybil and Collude attacks and can be used to reward workers for completing unverifiable tasks, while keeping the system free for the end-users. We use two exemplar use-cases to showcase the usefulness of PoP and we build a simulator to assess the cryptoeconomic behaviour of the system in terms of prestige transfer between nodes.

\end{abstract}

\section{Introduction}
Following the recent success of Bitcoin~\cite{bitcoin}, a plethora of cryptocurrencies  have experienced an increase of popularity~\cite{sok}. There are over 1,900 cryptocurrencies one can invest in with the total market cap exceeding 280B USD\footnote{\url{https://www.investing.com/crypto/currencies}}. 
In contrast to resource-wasting Proof-of-Work~(\pow) protocols~\cite{bitcoin} or Proof-of-Stake~(\pos) binding users' minting power to the amount of their coins~\cite{bentov2016snow, david2017ouroboros, kiayias2017ouroboros}, a recent trend sees cryptocurrencies as an incentive method for users to perform useful work and create a shared economy environment. For instance, Filecoin~\cite{combinatorjune,protocol2017por} rewards miners for renting storage capacity, while Golem~\cite{golem} allows to rent out processing power and perform client's computation-heavy tasks. The vision is to create a decentralised system, where miners are incentivised to do useful work, secure the transactions and automatically receive rewards when tasks are completed. 

In the classic setup, useful work can be performed by a \emph{``contributor"} for a \emph{``beneficiary"}. The beneficiary submits a task and a reward to the blockchain that is used to assure payment for service \cite{airtnt} \cite{krol2018spoc}. When the contributor correctly completes the requested task, the payment is unlocked. However, currently, multiple cloud platforms do not expect their users to pay for the services (\ie Facebook, Youtube). As a result, in order to attract more users, blockchain-based platforms must keep this featue as well. It means involving a third party in the system to whom we refer to as a \emph{``motivator"}. The motivator benefits from increasing the size and popularity of the network and rewards contributors' useful work, while keeping it free for the end-users. For instance, on Steem \cite{steem} authors (contributors) create content for readers (beneficiaries). Readers can then signal interesting content using a ``vote-up'' button. Steem (motivator), which is primarily interested in increased viewership in order to increase its income from advertisements, rewards the most successful authors with an automatic coin transfer. 

Such a system presents multiple benefits. Beneficiaries do not have to pay for the content, the system remains open for any contributor to join, perform useful work and be paid according to their performance and contribution; the motivator on the other hand benefits from an open platform avoiding costly contracts, and contributor selection process. However, for that to work, the network must be able to automatically verify tasks. While completion of some tasks can be proven to a third party (\eg file storage \cite{protocol2017por}), in many cases this is impossible. For instance, it is not possible to prove that a file has been successfully transferred between any two untrusted nodes. In such cases, the motivator relies only on beneficiary acknowledgments to reward contributors and can be thus susceptible to Sybil attacks. In order to maximize their reward, contributors can create multiple fake identities claiming usefulness of their work. Moreover, even with access restriction techniques or voting power bounded to stake, users can collude and cross-acknowledge their potentially non-existing work.

In this paper, we present Proof-of-Prestige~(\pop) that aims to be a building block of cryptocurrencies based on useful work. \pop builds on \pos, but instead of binding user minting power to stake, the probability of minting a new block is determined by users' prestige. One can generate \emph{prestige} associated with each identity in the system from money or by performing useful work; \emph{prestige} is much more volatile and renewable resource than virtual currencies. In \pop, beneficiaries pay for services by transferring prestige to contributors, while avoiding spending coins. A task is considered as completed when acknowledged by its beneficiary. \pop can thus be used in wide variety of use-cases, securing the system against Colluding and Sybil attacks and avoiding artificial inflation of miners' power by completing dummy tasks. \pop is not meant to be a new consensus protocol, but rather a new system of assigning minting power to users that is compatible with already existing \pos protocols. Proof-of-Prestige, does not deal with the fair exchange problem between beneficiary and benefactor~\cite{dziembowski2018fairswap}, but rather offers a secure way to reward benefactors, by network operators. Our system builds on Proof-of-Stake and can be easily integrated in a distributed ledger using the latter (\Cref{sec:goals}). We present two variants of our system, simple and progressive mining, and show how they can be used in real-world scenarios.

We make the following contributions:
\begin{itemize}
\item \Cref{sec:pop} presents Proof-of-Prestige~(\pop), a novel type of proof of useful work that can run on top of any Proof-of-Stake \blockchain, and that is resilient against Sybil and Collude attacks.
\item \Cref{sec:simple_mining} and \Cref{sec:progressive_mining} introduce two types of minting power distribution, simple and progressive mining, to support a wide variety of scenarios.
\item \Cref{sec:use_case} presents and analyzes two key use-case applications, a publisher platform and a file distribution system, that benefit from our framework.
\item \Cref{sec:eval} provides an extensive evaluation of our protocol and prove its security and high performance. 
\end{itemize}


\section{Threat Model and Goals} \label{sec:goals}

We assume the following actors in our system:
\begin{itemize}
    \item \textbf{Beneficiaries:} End-users of the system which transfer prestige to the contributors in exchange for a task to be performed (such as distributing a file).
    \item \textbf{Contributors:} Nodes which perform tasks for beneficiaries in return for prestige.
    \item \textbf{Motivators:} System operators or users that submit tasks to the network. Motivators do not benefit from services directly, but rather from expanded network size (\ie crowdsourcing their own tasks).
\end{itemize}

Each user is represented by an identity and can act as contributor, beneficiary and/or motivator. We assume that none of the actors can be trusted; \ie they may attempt to steal funds, avoid making payments, create fake transfers, and create fake identities. At any given time, each party may drop, send, record, modify, and replay arbitrary messages in the protocol.

\sysname assumes an underlying blockchain to facilitate \emph{prestige} transfers without a trusted third party. We assume that the underlying blockchain is resistant to double-spending attacks, and guarantees liveness - that is, transactions submitted to the blockchain will be eventually processed, within some defined period of time.

\sysname has the following design goals:
\begin{itemize}
    \item \textit{Open membership.} Any system user is able to participate in the system as a contributor or beneficiary without prior approval by some third party.
    \item \textit{Creditless rewards.} A beneficiary can reward a contributor for completing a task without actual credit (i.e., virtual currency), but rather by increasing the chance of the contributor minting next blocks.
    \item \textit{Contributor incetivization.} Contributors are economically incentivized to perform tasks.
    \item \textit{Inflation control.} Given a predefined rule to determine the inflation rate, the rate of inflation of coins in the system cannot be changed by users.
    \item \textit{Sybil attack resistance.} Creating multiple identities cannot increase a user's minting power without obtaining more coins.
    \item \textit{Collude attack resistance.} Users cannot increase their total amount of prestige by colluding with each others.
\end{itemize}

Furthermore, \sysname requires a \pos consensus protocol that can assign custom weights to users, where each weight determines the probability of minting a new block (i.e., Algorand \cite{gilad2017algorand}, Tendermint \cite{kwon2014tendermint}, Hashgraph \cite{baird2016swirlds}). As prestige is a volatile resource, it cannot be used with \pos protocols which require stakers to deposit and lock their coins for a long period of time (such as Casper FFG~\cite{casper}), so that their deposit can be reduced if they are caught misbehaving. This is because prestige cannot be locked, as by design it automatically increases or decreases. However, it is compatible with protocols such as Ouroboros \cite{kiayias2017ouroboros} which do not require deposits, and adapt to the new stake mapping in the system every epoch, where coins (and prestige) do not have to be locked up and can be transferred at any time.

\section{Proof of Prestige}\label{sec:pop}
In PoP, users have two values associated with their accounts:
\begin{itemize}
    \item \textbf{Coins} - similar to other crypto-currencies (such as Bitcoin or Ether) or ERC20 tokens \cite{eip20}. Coins can be directly sent to or received by other users via transactions, as is the case with all crypto-currencies. More importantly, \emph{coins generate prestige over time.}
    \item \textbf{Prestige} - determines the probability of the user minting a new block. Prestige cannot be directly transferred between accounts, but is exchanged as a reward for performing useful work. Contributors gain prestige by completing tasks, and beneficiaries lose prestige by acknowledging services.
\end{itemize}

Both values are related and can influence one another. The coins generate prestige over time up to a \emph{Static Value} (see \Cref{sec:prestige}), while prestige determines the probability of minting the next block, similarly to stake in PoS. \emph{The probability of each user minting a block is thus proportional to their prestige.} Minting a new block, in turn, allows users to collect transaction fees, discover new coins and thus increase their total amount of coins. Prestige represents a much more volatile and renewable resource, while the amount of coins does not change over time unless minted or transferred in transactions. \emph{Prestige can be spent to benefit from services or gained when performing services for others.}

There is one note worth making regarding the difference of \pop to \pos as regards the ability of rich users to gain more from the system: while in \pos the amount of coins (stake) a user has is the only resource that determines who mines the new block, in \pop useful work can increase a users' prestige and therefore, the probability of minting new coins, even if the user starts with a low amount of coins/prestige. Therefore, although someone who buys more coins can enter the system with higher prestige (and hence, higher probability of minting new coins), this does not exclude users with less coins from building up prestige if they contribute to the system with useful work.

\subsection{Prestige} \label{sec:prestige}
When  user $U_i$ joins the system (\ie creates their identity), they start with no prestige $P_i=0$. With each new block (mined by any user in the system) the amount of prestige of every user in the system is increased by the number of coins in their wallet $C_i$, so that richer users generate prestige at higher rate. At the same time, in order to avoid prestige increasing indefinitely, we introduce a decay parameter $d$. The decay determines what percentage of current prestige is lost with each block. Specifically, the prestige of a user evolves over time according to a non-homogeneous, autonomous, affine, first order, Discrete Dynamical System (DDS) with prestige increment on time-slot $t$:
\begin{equation}\label{eq:prestige_change}
    \delta P_i^{t} = C_i^{t} - dP_i^{t-1},
\end{equation}
 where $t\geq1$ and $0<d<1$ is a tunable system parameter. We can see that prestige on time-slot $t$ can be written as:
\begin{equation} \label{eq:prestige_sum}
P_i^t = \sum\limits_{j=1}^{t}\delta P_i^j = P_i^{t-1}+\delta P_i^{t} =C_i+(1-d)P_i^{t-1}.
\end{equation}The fixed point(s) of \Cref{eq:prestige_sum} DDS can be easily derived by applying simple linearisation techniques. In detail, consider that $P_i^t=g(P_i^{t-1})$, then for a candidate fixed point $S_i$ we have that $S_i=g(S_i)$: 
\begin{equation}\label{eq:static_value}
    S_i = \frac{C_i}{d}.
\end{equation} In fact, $|g'(S_i)|<1$, \ie $g'(\cdot)=1-d<1$, and therefore $S_i$ is an attractor fixed point indicating the convergence of prestige DSS to $S_i$. The amount of prestige that a user can generate from coins is thus limited to $S_i$---called \emph{Static Value}---where increasing decay evens up prestige generated from coins (\Cref{fig:prestige_evolution}). The static value depends only on the amount of coins in user's wallet $C_i$ and the decay parameter $d$ and it therefore increases by acquiring more coins.

To increase their prestige above their Static Value users have to perform useful work, confirmed by a beneficiary for whom the task was completed. When performing useful work, users instantly get more prestige (see Prestige spikes in \Cref{fig:prestige_evolution}) and therefore, increase their chances of minting the new block. In turn, minting a new block results in extra coins and thus, in higher \emph{Static Value} (see User $U_2$ in \Cref{fig:prestige_evolution}).

In \pop users do not pay with coins when benefiting from/receiving a service, but with prestige that, even if depleted, will be slowly replenished (as long as the user has some coins). However, when user prestige is higher than its static value, the decay exceeds prestige gained from coins and the user loses prestige until they reach their static value again. The reduction in prestige is another desired property, as it incentivises users to keep on contributing to the system by performing useful work. The network acts thus as a closed-loop control system correcting users current prestige to their static values (\Cref{fig:prestige_evolution}).

Our system needs to be secure and resistant against Sybil and Collude Attacks (\Cref{sec:goals}). When having a fixed amount of coins, we claim that users cannot gain more prestige by creating Sybil identities:

\begin{theorem}\label{th:prestige}
Each user has no prestige gain incentives to divide their coins into multiple identities.
\end{theorem}
\begin{proof}Let user $i$ divide their coins into $1,2,...,k$ identities according to $C_{i,1},C_{i,2},...,C_{i,k}$ respectively, such that $\sum\limits_{q=1}^kC_{i,q}=C_{i}$. Then the aggregated prestige of multiple identities at their limit will be:
\begin{align*}
S_{i,1}+S_{i,2}+...+S_{i,k} = \sum\limits_{q=1}^kC_{i,q}\slash d = C_i\slash d = S_i.
\end{align*}That is, the total prestige generated in the long-run by multiple identities equals the prestige achieved by a single identity. \end{proof}

\begin{figure}
  \includegraphics[width=1\linewidth]{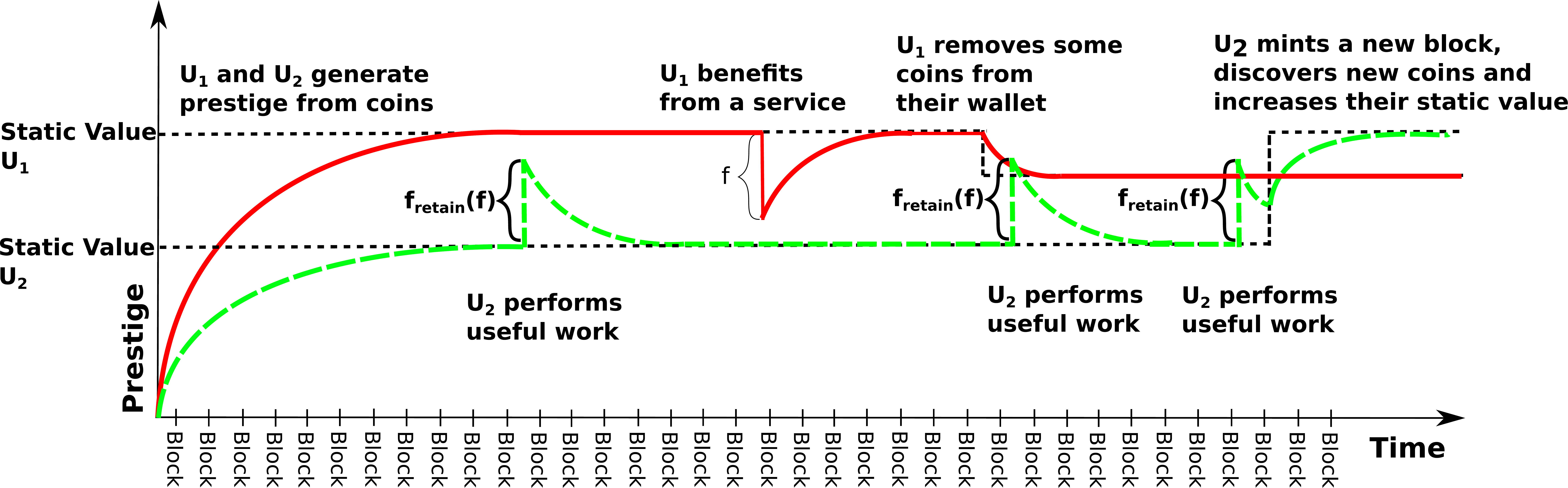}
  \caption{Example of evolution of the prestige and static value with time (expressed in blocks).}
  \label{fig:prestige_evolution}
\end{figure}

\subsection{Mining Overview} \label{sec:mining}
Users can act both as contributors and as beneficiaries and exchange services for prestige. The total change of prestige for useful work $\delta x_i^t$ of $U_i$ at block $t$ is thus given by prestige spent to benefit from services (being a beneficiary) and gained by performing useful work (being a contributor). 
\begin{equation}\label{eq:prestige_total}
  \begin{aligned}
    \delta x_i^t = x_{\textit{gained}}^t - x_{\textit{spent}}^t
  \end{aligned}
\end{equation}

The total amount of prestige that a user is spending is the sum of prestige spent on each service that $U_i$ benefited from $x_{spent}^i = \sum_{j=0}^{m}x_{ij}$, where the prestige fee transferred between each pair of nodes $x_{ij} = f$ can be predefined or negotiated between the beneficiary and the contributor. 

Analogically, the amount of prestige gained by a contributor is the sum of prestige gained on each service that $U_i$ performed, but is modified by a retain function $x_{gained}^i = \sum_{j=0}^{m}f_{\textit{retain}}(x_{ji})$. The retain function defines what percentage of received prestige will be kept by a contributor. We define different patterns of useful work (represented by simple and progressive mining) with different retain functions explained in detail in the following Sections.
When a task is completed by a contributor, the beneficiary recognizes it by generating a signed \emph{acknowledgment}, and sends it to the contributor (\Cref{sec:acknowledgments}). The contributor can then upload the acknowledgment to the blockchain to register the completed task and get the corresponding prestige. 

\begin{figure}
  \centering
  \includegraphics[width=1\linewidth]{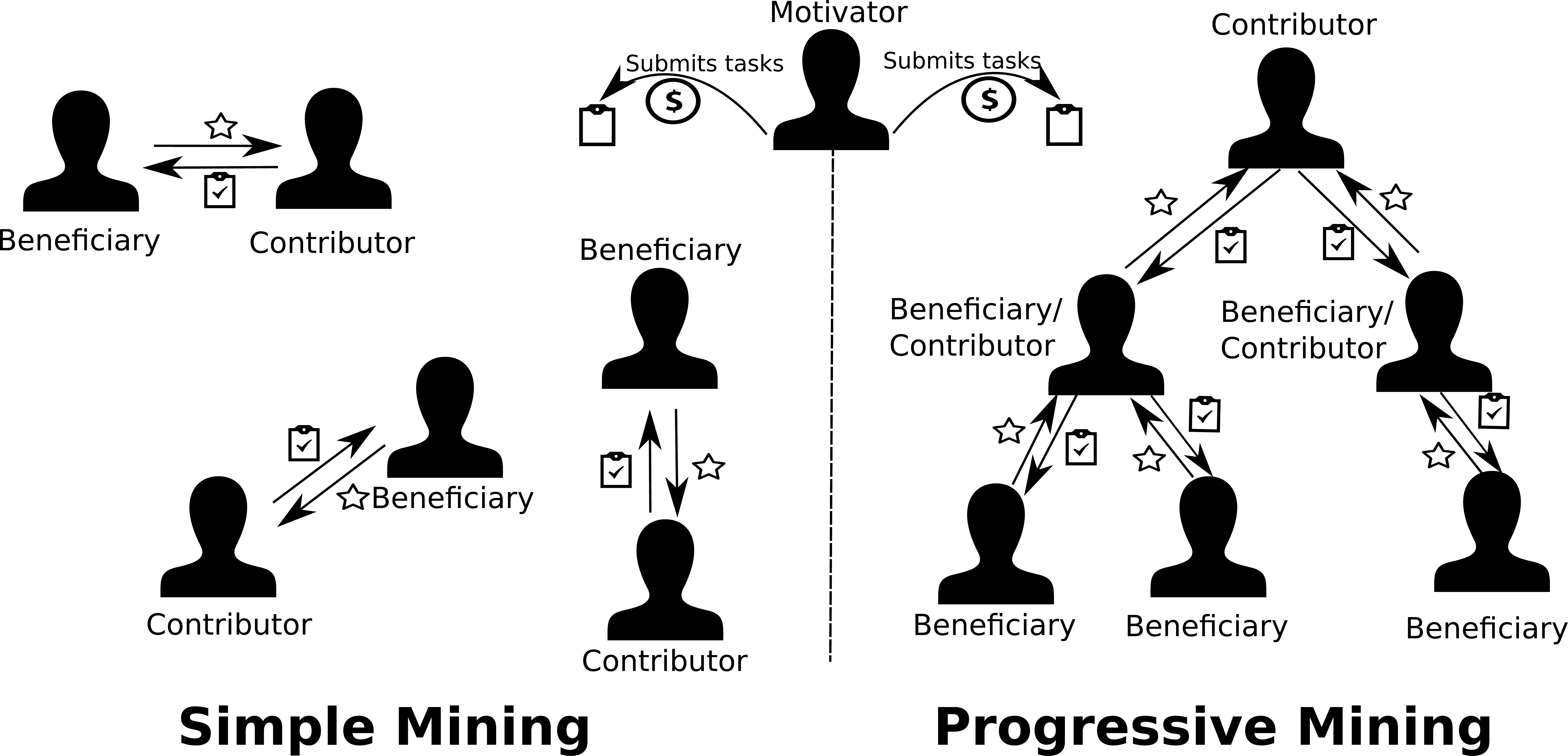}
  \caption{Simple and Progressive Mining.}
  \label{fig:simple_progressive}
\end{figure}

When performing useful work during block $k$ and as long as the amount of prestige transferred by the beneficiary $U_i$ is lower or equal to the prestige retained by the contributor $U_j$, so that $x_{ij}^k \geq f_\textit{retain}(x_{ji})^k$, then the aggregate amount of prestige possessed by $U_i$ \emph{and} $U_j$ does not increase. This means that \emph{users cannot increase their minting power by cross-acknowledging their work and the system is resistant against both the Sybil and Collude Attacks.} 

\begin{theorem}\label{th:mining}
There are no prestige gains produced by prestige transfers between users.
\end{theorem}
\begin{proof}We denote by $P_i^t$, $P_j^t$ the prestige of contributor $i$ and beneficiary $j$ when there is no prestige transfer from $i$ to $j$ and by  $\bar{P}_i^t$, $\bar{P}_j^t$ the prestige of contributor $i$ and beneficiary $j$ when the transfer $x_{ij}^k$ is taking place upon time-slot $t > k$. Then:
\begin{align*}
\bar{P}_i^t+\bar{P}_j^t&=P_i^{k-1}-x_{ij}^k+ \sum\limits_{q=k}^{t}\delta \bar{P}_i^q+P_j^{k-1}+x_{ij}^k+ \sum\limits_{q=k}^{t}\delta \bar{P}_j^q,\\
&=P_i^{k-1}+ \sum\limits_{q=k}^{t}\delta \bar{P}_i^q+P_j^{k-1}+ \sum\limits_{q=k}^{t}\delta \bar{P}_j^q,\\
&=P_i^{k-1}+C_i-d (P_i^{k}-x_{ij}^k)+\sum\limits_{q=k+1}^{t}\delta \bar{P}_i^q\\&+P_j^{k-1}+ C_j-d (P_j^{k}+x_{ij}^k)+\sum\limits_{q=k+1}^{t}\delta \bar{P}_j^q,\\
&=P_i^{k}+\sum\limits_{q=k+1}^{t}\delta \bar{P}_i^q+P_j^{k}+\sum\limits_{q=k+1}^{t}\delta \bar{P}_j^q,\\
\vdots\\
&=P_i^t+P_j^t.\\
\end{align*}That is, prestige transfers do not affect the total prestige that exists in the system.
\end{proof}
The result of \Cref{th:mining} can be easily extended for the general case of $N$ users and all time-slots.

\subsection{Simple Mining} \label{sec:simple_mining}
We define \emph{Simple Mining} to reward services performed uniquely between one contributor and one beneficiary (i.e., renting out contributor's CPU power to perform beneficiary's computations). In \emph{Simple Mining}, when an acknowledgment of service performed by contributor $U_i$ for beneficiary $U_j$ is submitted to the blockchain, the contributor retains the whole amount of transfered prestige $f$ so that $f = x_{ij}^t = -x_{ji}^t$ (\Cref{fig:simple_progressive}). Therefore, the simple mining retain function is defined as:

\begin{equation}\label{eq:retain_simple}
  \begin{aligned}
    f_{\textit{retain}}(x_{ij}) = x_{ij}
  \end{aligned}
\end{equation}

In \emph{Simple Mining}, the retained value is equal to the transferred value. Therefore, \Cref{th:mining} applies and we conclude that, \emph{Simple Mining} is resistant to Sybil and Collude Attacks.

\subsection{Progressive Mining} \label{sec:progressive_mining}
For cases where benefiting from a service, allows the beneficiary to perform useful work for others (e.g., seeding in a content distribution system), we introduce the concept of \emph{Progressive Mining}. In \emph{Progressive Mining, contributors are rewarded by their own useful work, but also for work performed by their beneficiaries.} That is, if $U_i$ performs a service for $U_j$, $U_i$ will receive some prestige for each service performed by $U_j$ and other nodes that $U_j$ provided the service to (\Cref{fig:simple_progressive}). The scheme can be seen as a Directed Acyclic Graph (DAG) with users as nodes and edges representing useful work performed for subsequent users. In this case, \emph{Prestige is flowing from the leaves towards the root.}

\vspace{-5pt}
This type of rewarding scheme is useful in scenarios such as file propagation, where receiving a file allows to distribute it to other users. At the same time, we want to protect the system from distribution manipulations that would change the amount of prestige received by legitimate parties.

In particular, when $U_j$ performs a task for $U_i$, the transferred prestige value $x_{ji}$ is not directly added to $U_j$'s account. Instead, $U_j$ must share earned prestige with its DAG predecessors and can retain only a part of earned prestige:
\begin{equation}\label{eq:progressive_transfer}
  \begin{aligned}
    f_{\textit{retain}}(x_{ji})  = \frac{x_{ji}*P_i}{P_i + P_{b}^i}
  \end{aligned}
\end{equation}
where $P_{b}^i$ is branch power of $U_i$ expressed as the sum of the prestige values of the predecessors of $U_i$ multiplied by a branch power parameter $b$:
\begin{equation}\label{eq:branch_power}
  \begin{aligned}
    P_{b}^i  = \textit{sum\_prestige}(predecessors(U_i)) * b
  \end{aligned}
\end{equation}
Such a branch power function incentivises users to attach to shorter branches, automatically balancing the DAG. 
$U_i$ will thus keep only a part of $x_{ji}$, while the remaining part will be transferred upstream towards the root. If $P_b^i=0$, no prestige is sent upstream (which is the case for the DAG root) and users with no base prestige cannot retain any prestige flowing upstream, which protects the scheme from DAG manipulations using Sybil identities. With increasing $P_b^i$ more prestige is pulled upstream towards the root of the distribution tree.

The whole scheme is based on user's own prestige and the prestige sum of its predecessors. It is fully resistant to topology manipulation using Sybil identities that do not have any prestige. Users also cannot increase their prestige gain by spreading their coins (and thus gain prestige) over several artificial identities. 

\begin{theorem}In Progressive Mining, users cannot retain more prestige by splitting their coins into multiple identities.
\end{theorem}
\begin{proof}Without loss of generality assume that user $i$ divides their coins into 2 identities according to $C_{i,1}$ and $C_{i,2}$, such that $C_{i,1}+C_{i,2}=C_{i}$. Hence, the prestige retained by the multiple identities of user $i$ out of the prestige transferred from user $j$, $x$, will be:
\begin{align*}
f_{retain,1}(x)-x+f_{retain,2}(2x-f_{retain,1}(x))=\\=x\frac{P_{i,1}}{P_{i,1}+bP_{i,2}+P_b^i}-x\\+x\left(2-\frac{P_{i,1}}{P_{i,1}+b P_{i,2}+P_b^i}\right)\frac{P_{i,2}}{P_{i,2}+P_b^i},\\
=x\frac{P_{i,1}P_{i,2}+P_{i,2}P_{b}^i+bP_{i,2}^2-bP_{i,2}P_{b}^i-(P_{b}^{i})^2}{(P_{i,1}+b P_{i,2}+P_b^i)(P_{i,2}+P_b^i)},\\
\leq x\frac{P_{i,1}+P_{i,2}}{P_{i,1}+P_{i,2}+P_b^i},\\
=x\frac{P_{i}}{P_{i}+P_b^i},\\
=f_{retain(x)}.
\end{align*}where the inequality applies since:
\begin{align*}
\frac{P_{i,1}P_{i,2}+P_{i,2}P_{b}^i+bP_{i,2}^2-bP_{i,2}P_{b}^i-(P_{b}^{i})^2}{(P_{i,1}+b P_{i,2}+P_b^i)(P_{i,2}+P_b^i)}\\\leq\frac{P_{i,1}+P_{i,2}}{P_{i,1}+P_{i,2}+P_b^i}\Longleftrightarrow
\\-P_{b}^i\left( bP_{i,1}P_{i,2}+P_{i,1}P_{b}^i+P_{i,2}(bP_{i,2}+P_{b}^i+bP_{b}^i)+(P_{b}^i)^2\right)\\\leq P_{i,1}(P_{i,1}P_b^i+bP_{i,2}P_b^i+(P_b^i)^2).
\\
\end{align*}On the other hand, from Theorem 1 we know that multiple identities have no prestige gains, \ie $P_{i,1}+P_{i,2}=P_i$, and therefore the third equality is valid. Hence, users cannot increase their prestige by creating multiple identities in progressive mining, intuitively due to the prestige payments that is forced to submit to their fake identities that in turn they retain only a portion of the prestige.

\end{proof}

\subsection{Acknowledgments}\label{sec:acknowledgments}
Nodes generate \emph{Acknowledgments} when benefiting from useful work. Contributors earn prestige by submitting a received \emph{Acknowledgment} to the \blockchain, showing that they provided some service to other nodes. 

For \emph{Simple Mining} (\Cref{sec:simple_mining}), acknowledgments are standard digital signatures. The beneficiary generates a signature on an ID uniquely identifying the task, the contributor's public key, and the agreed amount of prestige to transfer. This is transferred to the contributor who uploads it to the \blockchain to trigger the prestige transfer.

For \emph{Progressive Mining} (\Cref{sec:progressive_mining}), acknowledgments are composite signatures~\cite{saxena2014increasing}. Each node in the DAG branch composes its signature with the initial signature generated by the DAG root, forming a composite signature that contains the signature of each node involved in the task. 

When a new task is added, the DAG root node performs services to beneficiaries and collects their signatures $\sigma_{i,ID}$ over an ID uniquely identifying the task, the contributor's public key, and the agreed amount prestige to transfer. The root then submits the signed message to the \blockchain to receive prestige, while beneficiaries can turn into contributors and continue performing services for other users. These nodes start by sending $\sigma_{i,ID}$ to potential recipients. Each recipient checks the validity of the signatures and the task can be performed if the check passes. A beneficiary $U_j$ generates their own signature, and composes it with the previous signature $\sigma_{i,ID}$, obtaining a composite signature $\sigma_{j,ID}$. The beneficiary then sends the resulting composite signature to the contributor who uploads it to the \blockchain in order to update the prestige value of the previous contributors' and the root node. The above process continues as the DAG grows.

Before accepting a service using progressive mining, the beneficiary should verify that the contributor is already included in the DAG. In order to achieve this, the contributor transmits their own composite signature indicating the path from the DAG root to itself---this allows the beneficiary to register the transaction on the \blockchain even if their predecessors do not do it. 
The properties of composite signatures prevent any single or subset of signature(s) from being removed from the composite~\cite{saxena2014increasing}. Furthermore, the system will update the prestige of all nodes even if only the last sender in each branch uploads the composite signature to the \blockchain. However, it is in the best interest of every contributor to upload the composite signature, in case no other subsequent node uploads theirs. 

After benefiting from a service, the beneficiary might refuse to send back the acknowledgment, or might attempt to generate an acknowledgment for another, colluding node. While multiple solutions exist to ensure fair exchange between two mutually distrusting parties~\cite{airtnt, pagnia1999impossibility, dziembowski2018fairswap}, a specific solution should be tailored to the nature of performed tasks. We thus sketch some solutions in Sec. \ref{sec:use_case} and leave more detailed discussion for future work.

\subsection{Prestige Economics}
In the previous sections, we explained the prestige flow between users. However, prestige is a volatile resource and does not represent real assets. High prestige value increases the chances of a node to be elected to submit a new block. That said, in order to incentivise users, it must be bound to a reward expressed in coins. In current blockchain payments systems, such as Bitcoin, the rewards come from \first fees paid by users submitting transactions included in the block and \second new coins being discovered with each new block. The fees protect from Denial of Service (DoS) attacks, but also increase the exploitation cost of the system. On the other hand, new coins increase inflation and de facto reward the successful miner from money stored in the wallets of other users. 

\pop is compatible with both reward methods described above, but additionally, we introduce a third type of reward based on optional fees paid by motivators. Motivators directly benefit from increasing network size and can add coins as an additional reward for mining new blocks. The reward can be specified as an amount of coins distributed per block within limited duration (i.e., 1000 blocks). Motivators can thus incentivise users to participate in a specific task when their services are needed the most. 

\section{Use Cases} \label{sec:use_case}
We present two applications that leverage  \pop to support a reliable and secure incentive mechanism: a Publisher Platform as a system that uses simple mining, and a File Distribution system as a system that needs progressive mining.

\subsection{Simple mining - Publisher Platform}
Steem~\cite{steem} is a popular publishing platform that uses a combination of Proof-of-Brain and Delegated Proof-of-Stake~\cite{dpos} to reward content creators (contributors) and readers (beneficiaries). Readers can vote for interesting content using their Steem Power acquired by committing their coins to a thirteen-week vesting schedule. The best creators and readers are then rewarded by newly created coins. Steem is thus resistant to Sybil attacks (fake identities do not have coins and thus cannot acquire Steem Power), but is vulnerable to collude attacks. As users can recommend unlimited number articles, they are incentivized to cross-recommend their work with other users in order to maximize their reward. 

Introducing our \pop with Simple Mining into Steem could solve this problem. Replacing Steem Power with prestige, introduces a limit on the number of content items that readers can vote for within a specified period of time. Increased amount of coins replenishes prestige faster so richer users have still higher voting power as in the current scheme. Cross-recommendations no longer make sense as they do not increase the combined reward of the involved parties. Readers are thus incentivized to vote for only the best articles they find in order to receive high quality content in the future.

\subsection{Progressive mining - File Distribution}\label{sec:use_case_progressive}
We base our File Distribution use case on already existing systems requiring a secure rewarding scheme for file propagation such as Filecoin/IPFS~\cite{combinatorjune} or NOIA~\cite{noia}. Currently those systems either do not provide incentives for file propagation (IPFS), or are vulnerable to Sybil attacks (NOIA). A content creator (acting as a motivator) wants to distribute and pre-fetch large video files among mobile phone users at the edge of the network in order to increase its viewership. In the current/traditional client-server model, content is pulled from the Content Delivery Network (CDN) server upon the user’s request. Using CDNs, however, involves substantial fees and scales badly when crossing the mobile data link. In our File Distribution system, users themselves contribute to the distribution of the content directly between mobile devices without involving the network or third party CDNs. Effectively, users participate in an \emph{``incentivised content propagation network"}, where they get rewarded for contributing their resources to the network. At the same time, content creators/publishers can significantly reduce the cost of content distribution.

In the beginning, a file is directly transferred from the creator to a few users initiating a distribution DAG with the creator as its root (right part of \Cref{fig:simple_progressive}). The propagation then continues directly between user devices expanding the DAG. Each time a user receives a file, they acts as beneficiary and must generate an acknowledgment for the contributor who sent the file. In order to reduce the risk of a malicious beneficiary walking away without generating an acknowledgment, each file is partitioned into small chunks. A new chunk is sent only if an acknowledgment for the previous one has been received. 

This type of content propagation is a typical use-case for progressive mining: 
a contributor transmits a file to a beneficiary and is rewarded not only for this one-hop transfer, but for all the subsequent transfers in their subtree.

\section{Evaluation} \label{sec:eval}
In order to evaluate the behaviour of Proof-of-Prestige, we developed a Python3 simulator\footnote{\url{https://gitlab.com/mharnen/pop}}. In the following, we investigate a number of factors that influence users' prestige fluctuation and present one scenario for each of those factors.


\textbf{Decay Parameter, $d$:} We start by investigating the behaviour of the prestige correction function shown in \Cref{eq:prestige_change}. \Cref{fig:eval_decay} shows the influence of the decay parameter $d$ on prestige gain. We create 4 users with different amounts of coins and $d$ values. 
All the users start with zero prestige $P^0=0$; we add prestige $\delta P_i=200$ at block $t=100$ to each user, and remove prestige by $\delta P=-200$ at block $t=150$. The number of coins determines the static value (\Cref{eq:static_value}), but the number of coins does not influence the time needed to converge back to the static value. On the other hand, increasing the decay parameter $d$ lowers the static value and reduces the time required by each user to reach their static value from $P=0$. The same applies when users receive ($t=100$) or spend ($t=150$) prestige. A higher decay parameter will make prestige go back to their static values faster. In contrast, reducing the decay parameter increases the value of prestige gained from useful work in comparison to prestige gained from coins.

\textbf{Gained Prestige:} \Cref{fig:eval_sum_decay} introduces 4 users with the same amount of coins $C = 100$ and prestige set to the corresponding static value $P_i = S_i$. We then inject different amounts of prestige per block to each user, reflecting the case where each user has provided services of different value, and let the system run for $10,000$ blocks. We measure the sum of gained prestige for this period for different values of the decay parameter $0<d<1$ (x-axis). The impact of gained prestige (and thus of useful work) increases exponentially for small values of $d$, while it decreases for higher values of $d$ (notably for $d\approx1$, where all the gained prestige is removed in the next block). 

\begin{figure*}[!ht]
\begin{minipage}[t]{.33\textwidth}
  \includegraphics[width=1\linewidth]{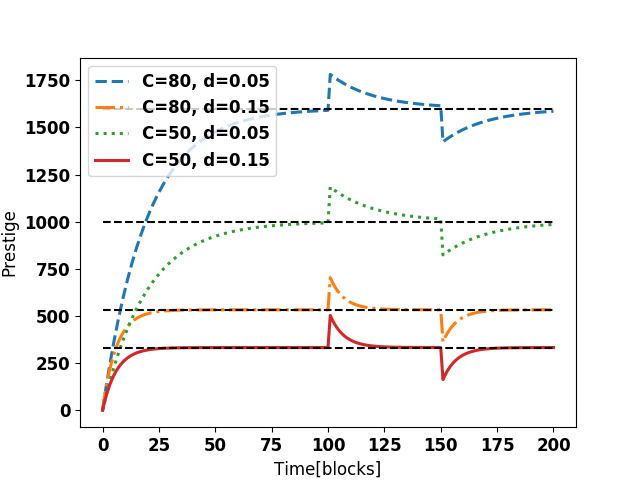}
  \caption{\scriptsize {Prestige over time.}}
  \label{fig:eval_decay}
\end{minipage}%
\begin{minipage}[t]{.33\textwidth}
  \centering
  \includegraphics[width=1\linewidth]{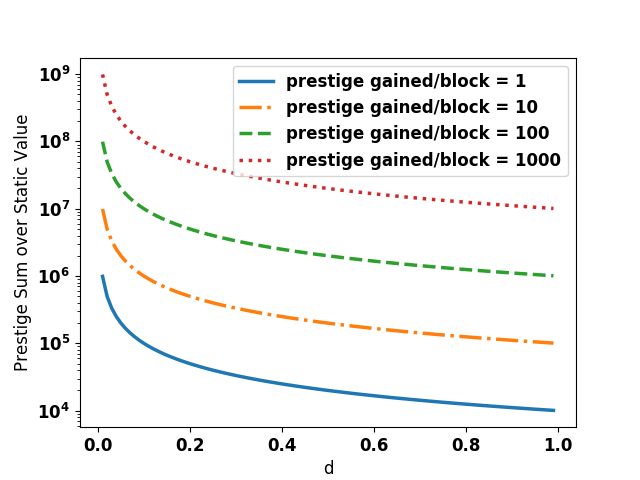}
  \caption{\scriptsize {Prestige sum above static value.}}
  \label{fig:eval_sum_decay}
  \end{minipage}%
\begin{minipage}[t]{.33\textwidth}
  \centering
  \includegraphics[width=1\linewidth]{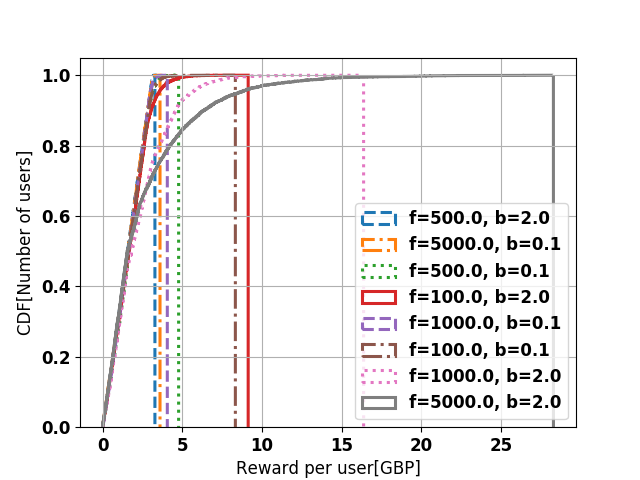}
  \caption{\scriptsize {Reward Distribution.}}
  \label{fig:bbc_rewards}
  \end{minipage}%
  \vspace{-15pt}
\end{figure*}



\textbf{Distance from the DAG root:} We create 1000 users involved in 100 random DAGs, where each edge represents performing useful work between nodes. Initial prestige values follow a uniform distribution from 0 to 100, while the branch power parameter is set to $b = 0.5$. \Cref{fig:eval_distance} presents the average prestige gained by each user and standard error as a function of distance from the DAG root. In simple mining, the root gains significantly more prestige than other users as it acts only as a contributor and does not benefit from (and thus pay for) services. For the other nodes, the distance does not influence the prestige gain and the standard error is low.

In contrast, in progressive mining, users located close to the root have the highest average prestige gain, while with increasing distance, the rewards decrease. This is a desired behaviour, as users close to the root have larger subtrees and collect fees from useful work of their successors. Progressive mining takes into account multiple factors when calculating the reward (\eg base prestige, distance fom the root) which results in increased standard error. 

\textbf{Base Prestige:} We investigate the prestige gain as a function of base prestige (\Cref{fig:eval_prestige}). Simple mining is not influenced by base prestige, while in progressive mining, with increasing base prestige, a user can collect higher rewards. This mechanism prevents the Sybil attack and rewards users who invested more in the system. However, a small fraction of high base prestige users experience prestige loss. Those are users that in spite of having high prestige, benefit from services and do not perform any useful work, which is another desired behaviour.

\textbf{Number of Completed Tasks:} Next, we investigate the effect of useful work to user prestige gain (\Cref{fig:eval_transfers}). For both progressive and simple mining, the average gain increases linearly with the number of services performed for other users. This experiment proves that for both mining modes, users with low base prestige and located further from the root in the distribution tree, can gain significant amounts of prestige by being useful for the network.

\begin{figure*}[!ht]
\begin{minipage}[t]{.33\textwidth}
  \centering
  \includegraphics[width=1\linewidth]{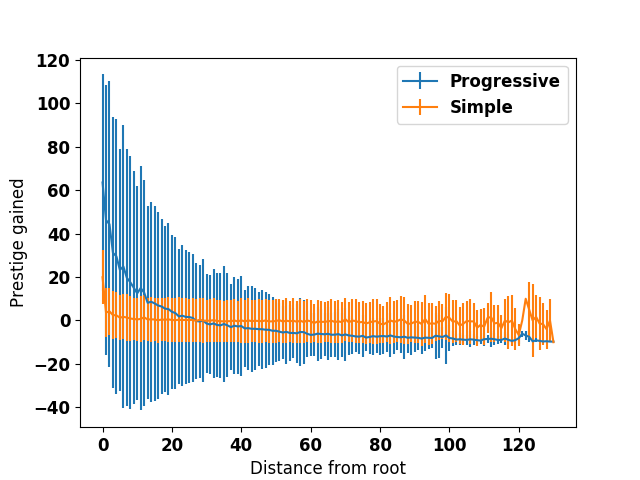}
  \caption{\scriptsize {Distance from root}}
  \label{fig:eval_distance}
\end{minipage}%
\begin{minipage}[t]{.33\textwidth}
  \centering
  \includegraphics[width=1\linewidth]{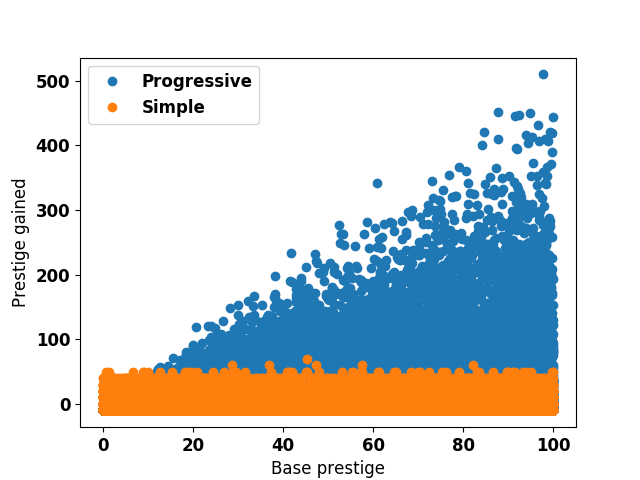}
  \caption{\scriptsize {Base Prestige}}
  \label{fig:eval_prestige}
  \end{minipage}%
\begin{minipage}[t]{.33\textwidth}
  \centering
  \includegraphics[width=1\linewidth]{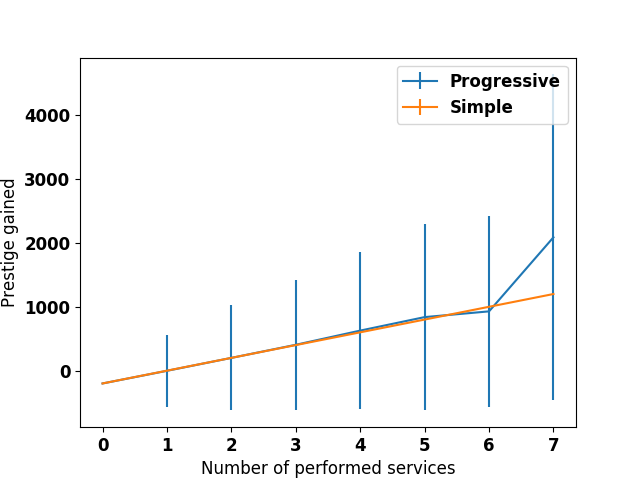}
   \caption{\scriptsize {Useful work}}
  \label{fig:eval_transfers}
  \end{minipage}%
  \vspace{-15pt}
\end{figure*}

\textbf{Contributor Involvement Probability:} We investigate prestige dynamics over time with both simple (\Cref{fig:eval_global_simple}) and progressive (\Cref{fig:eval_global_prog}) mining to provide a global view of the system. We introduce different poor ($C=50$) and rich ($C=100$) users having  $W=5\%$ or  $W=20\%$ probability of performing useful work with each new block in a random DAG containing 100 nodes, service fee $f=200$, and decay parameter $d=0.05$. At the beginning of the simulation, prestige values of each user go from $0$ towards their corresponding static value. In simple mining, users gain steady and moderate amounts of prestige and we do not observe differences between poor and rich users. In comparison, users performing progressive mining can reach much higher prestige gains and increased amount of coins (and thus base prestige) allows richer users to maximize their reward. The amount of performed useful work $W$ is important in both schemes, but its impact is higher in progressive mining, where poor, but active users, can reach prestige values similar to much richer, but less active nodes. 



\textbf{Contribution vs Coins Tradeoff:} We conclude by investigating the total value of acquired prestige over $1,000$ blocks by rich ($C=50$) and poor ($C=10$) users having different probabilities of performing useful work ($W=5\%$ and $W=25\%$) for different values of the parameter $d$ (\Cref{fig:eval_end}). For small values of $d$, useful work is more important than money. Poorer, but more active users, can thus acquire substantial amount of prestige, eventually surpassing rich users. This effect is decreased with increased values of $d$. On the other end of the spectrum, for high values ($d > 20$) the sum of acquired prestige depends mostly on user money, and is independent from the amount of performed useful work.  

\begin{figure*}[!ht]
\begin{minipage}[t]{.33\textwidth}
  \centering
  \includegraphics[width=1\linewidth]{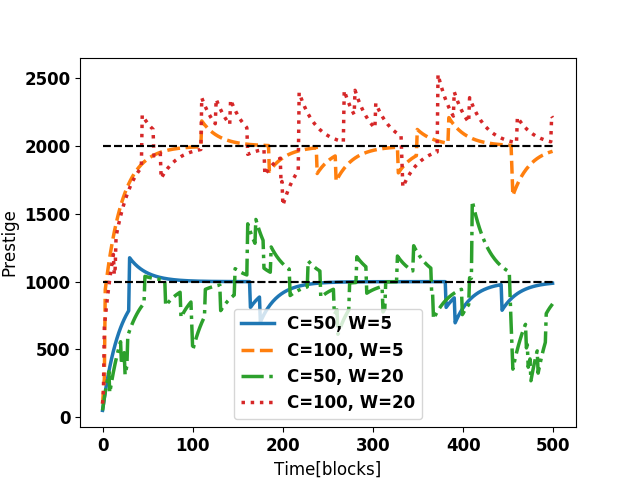}
  \caption{\scriptsize{Prestige evolution for work Simple Mining}}
  \label{fig:eval_global_simple}
\end{minipage}%
\begin{minipage}[t]{.33\textwidth}
  \centering
  \includegraphics[width=1\linewidth]{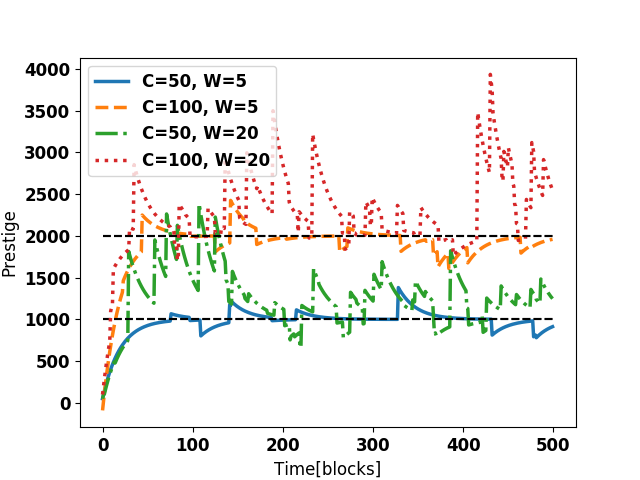}
  \caption{\scriptsize{Prestige evolution for work Progressive Mining}}
  \label{fig:eval_global_prog}
\end{minipage}%
\begin{minipage}[t]{.33\textwidth}
  \centering
  \includegraphics[width=1\linewidth]{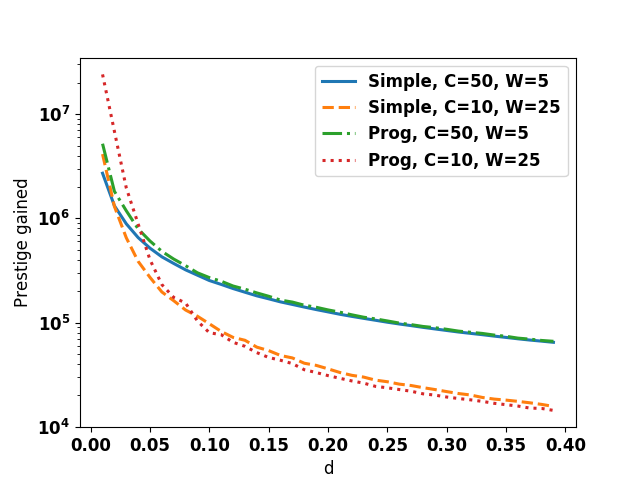}
  \caption{\scriptsize {Prestige sum acquired by different users.}}
  \label{fig:eval_end}
\end{minipage}%
\vspace{-10pt}
\end{figure*}
\textbf{Rewards:} To investigate potential rewards for users using Proof-of-Prestige, we focus on the File Distribution use case presented in \Cref{sec:use_case_progressive} and apply it to popular a BBC Series - \textit{Bodyguard}. BBC does not include advertisement in their content --- thus, each consecutive download increase the broadcaster's cost. With the number of viewers ranges from 14M to 17M\footnote{\url{https://www.barb.co.uk/viewing-data/four-screen-dashboard/}}, and the file size of approximately 250MB, the cost of delivering one series season using a CDN equals 4.7M\$\footnote{\url{http://cdncomparison.com/}}. For each episode of the series, we create a DAG with the corresponding size of users with random amount of base prestige ranging from 1 to 10000 and distribute the money spent on CDN to users proportionally to their prestige after performing useful work. In this scenario users transfer files to maximum 8 users. \Cref{fig:bbc_rewards} shows the reward distribution for different  values of $f$ and $b$. Surprisingly, those parameters have a negligible effect on the majority of the nodes. With changing parameters, user receive, but also transfer upstream different amounts of prestige. Such a behaviour influences mainly the most active nodes located close to the DAG root. With high $f$ and $b$ those nodes can acquire significantly higher amounts of prestige and thus collect higher rewards. The most active users collect up to 30\$ reward, while the distribution remains free for all the users. 

\textbf{Volume of Data Submitted to the Blockchain:} Finally, we approximate the amount of data submitted to the blockchain, which varies between simple and progressive mining. Simple mining requires a separate acknowledgment for each interaction---the number of submitted ACKs equals the number of performed tasks. The size of an ACK for simple mining is about 102 bytes; it is computed as the sum of the size of the composite signature\footnote{We implement composite signatures with BGLS signatures~\cite{bgls}.} (33 bytes), the task ID (32 bytes), the contributor public key (33 bytes), and the amount of prestige transferred (set to 4 bytes). Progressive mining requires an acknowledgment for each leaf in the DAG since composite signatures contain information about all the nodes between the root and the leaf; and the size of the ACK increases linearly with the depth of the DAG. The ACK size is therefore, $(33+69\times n_i)$ bytes, where $n_i$ is the depth of the DAG following path $i$. While this represents a substantial amount of data, acknowledgments can be processed off-chain using platforms such as Plasma~\cite{poon2017plasma} and update users' prestige periodically, significantly decreasing the volume of information kept on the main-chain.





\section{Related Work}

There are currently multiple systems focusing on rewarding miners for useful work. The largest group focuses on proving file storage where a verifier sends a file to a prover and later requests a proof that the prover really stored the file \cite{golle2002cryptographic, bowers2009proofs, di2003lkhw, ateniese2007provable, juels2007pors, protocol2017por, maidsafe, miller2014permacoin}. Alternatively, several platforms allow the prover to convince that the prover has access to some space \cite{honsi2017spacemint} \cite{chia} \cite{perito2010secure} \cite{dziembowski2011one} \cite{karvelas2014efficient} \cite{ateniese2014proofs}. Additionally one can require that the proof implies that the space also was erased \cite{perito2010secure},\cite{karvelas2014efficient}, or some function can only be computed in forward direction \cite{dziembowski2011one}. In all those cases, the network must be able to reliably verify the completed task to reward miners which narrows the scope of supported tasks to a small group. Some solution replace traditional Proof-of-Work with useful mathematical tasks that are easy to verify such as poynomials evaluation\cite{ball2017proofs, hu2015secure}. However, such systems support only one type of tasks and does not accept custom ones requested by users. 

Another family supports broader range of tasks (such as custom computation tasks)\cite{airtnt} \cite{krol2018spoc} \cite{iexec} \cite{chen2017security}, but relies on Trusted Execution Environments (TEEs) and Remote Attestation Protocols \cite{costan2016intel} \cite{winter2008trusted} to verify that the computations are being run on a genuine platform and the results are correct. However, TEEs are not available on every platform, require users to trust hardware vendors and are susceptible to side channel attacks \cite{gotzfried2017cache} \cite{weichbrodt2016asyncshock}. In contrast, \pop does not make any assumptions on users hardware. 
Some open platforms rely on centralized 3rd parties acting as consents to validate tasks or perform conflict resolution \cite{golem} \cite{sonm} \cite{bounty0x}. However, for those system to work correctly, the 3rd parties must be trusted by all network participants. Such an approach contrasts with the idea of open and trustles system lying behind blockchain and exposes the network to multiple colluding attacks. 

Finally, there exists multiple industrial projects in which the network operator (motivator) wants to reward users for performing useful work that is difficult or impossible to prove to a 3rd party making those system susceptible to Sybil attacks. Such tasks include content creation \cite{steem}, content distribution \cite{noia} or providing hardware for game players \cite{playkey}. The need for a blockchain-based reputation system/reward scheme has also been expresed for multiple systems ranging from the cloud \cite{habib2010cloud, hwang2010trusted, resnick2002trust} to IoT devices \cite{huckle2016internet, sun2016blockchain}. \pop can be a valuable addition to those system allowing to reliably reward users for truly performed tasks.

\section{Conclusion}
We presented \emph{Proof-of-Prestige} (PoP)---a reward system that can run on top of any Proof-of-Stake blockchain. We introduce the notion of \emph{Prestige} that is a volatile and renewable resource, is generated from coins and useful work, and can be spent to benefit from services. In \pop, each user's probability of minting a new block is directly determined by their prestige. 

In contrast to \pos, where the amount of coins (stake) a user has is the only resource that determines who mines the new block, \pop allows the network to reward contributors for their useful work acknowledged by beneficiaries. Our scheme is resistant to Sybil and Collude attacks and can be used in multiple scenarios without requiring to prove task completion to the network. Rather, a task is considered to be completed once confirmed by its beneficiary. 

We presented two variants of our scheme---simple and progressive mining, and showed how they can be used in real-world scenarios. Our evaluation confirmed that within both schemes users with low amounts of coins who contribute to the network can acquire significant amounts of prestige, similar to rich and ``lazy'' users. \pop reduces inequalities present in \pos and incentivises users to perform useful work.

\bibliographystyle{ieeetr}
\bibliography{pop.bib}

\end{document}